\documentclass[letterpaper, 10 pt, conference]{ieeeconf}
\IEEEoverridecommandlockouts
\usepackage{amsfonts,amsmath,amssymb} % assumes amsmath package installed

\DeclareMathOperator{\diag}{diag}         % Define diag operator
\def\rbb{\mathbb{R}}
\def\cbb{\mathbb{C}}
\def\trp{^T}
\def\diag{{\rm diag}}
\def\half{\frac{1}{2}}
\usepackage{psfrag,color}
\usepackage{enumerate,cite,latexsym,graphicx}
\newtheorem{theorem}{Theorem}
\newtheorem{lemma}{Lemma}

\newtheorem{definition}{Definition}

\newtheorem{remark}{Remark}
\title{\LARGE \bf Implementation of a Distributed Coherent Quantum Observer}

\author{Ian R.~Petersen and Elanor H. Huntington
\thanks{This work was supported by the Air Force Office of Scientific
Research (AFOSR). This material is based on research sponsored by the
Air Force Research Laboratory, under agreement number FA2386-16-1-4065.  The U.S. Government is authorized to reproduce and
distribute reprints for Governmental purposes notwithstanding any
copyright notation thereon.
The views and conclusions contained herein are those of the authors
and should not be interpreted as necessarily representing the official
policies or endorsements, either expressed or implied, of the Air
Force Research Laboratory or the U.S. Government. This work was also supported by the
Australian Research Council (ARC). }%
\thanks{Ian R. Petersen and Elanor H. Huntington are with the 
Research School of Engineering, The Australian National University, Canberra, ACT 0200,
Australia. Email: i.r.petersen@gmail.com, Elanor.Huntington@anu.edu.au.}
}%

\begin{document}

\maketitle
\thispagestyle{empty}
\pagestyle{empty}

\begin{abstract}
This paper considers the problem of implementing a previously proposed  distributed direct coupling quantum observer for a closed linear quantum system. By modifying the form of the previously proposed observer, the paper proposes a possible experimental implementation of the observer plant system using a non-degenerate parametric amplifier and a chain of optical cavities which are coupled together via optical interconnections. It is shown that the distributed observer converges to a consensus in a time averaged sense in which an output of each element of the observer estimates the specified output of the quantum plant.
\end{abstract}

%%%%%%%%%%%%%%%%%%%%%%%%%%%%%%%%%%%%%%%%%%%%%%%%%%%%%%%%%%%%%%%%%%%%%%%%%%%%%%%%
\section{Introduction} \label{sec:intro}
In this paper we build on the results of \cite{PET14Ca} by providing a possible experimental implementation of a direct coupled distributed quantum observer. A number of papers have recently considered the problem of constructing a coherent quantum observer for a quantum system; e.g., see \cite{VP9a,MEPUJ1a,MJP1}. In the coherent quantum observer problem, a quantum plant is coupled to a quantum observer which is also a quantum system. The quantum observer is constructed to be a physically realizable quantum system  so that the system variables of the quantum observer converge in some suitable sense to the system variables of the quantum plant. The papers \cite{PET14Aa,PET14Ba,PET14Ca,PET14Da}  considered the problem of constructing a direct coupling quantum observer for a given quantum system. 

In the papers  \cite{MJP1,VP9a,PET14Aa,PET14Ca}, the quantum plant under consideration is a linear quantum system. In recent years, there has been considerable interest in the modeling and feedback control of linear quantum systems; e.g., see \cite{JNP1,NJP1,ShP5}.
Such linear quantum systems commonly arise in the area of quantum optics; e.g., see
\cite{GZ00,BR04}. In addition, the papers \cite{PeHun1a,PeHun2a} have considered the problem of providing a possible experimental implementation of the  direct coupled observer described in \cite{PET14Aa} for the case in which the quantum plant is a single quantum harmonic oscillator and the quantum observer is a single quantum harmonic oscillator. For this case, \cite{PeHun1a,PeHun2a} show that a possible experimental implementation of the augmented quantum plant and quantum observer system may be constructed using a non-degenerate parametric amplifier (NDPA) which is coupled to a beamsplitter by suitable choice of the amplifier and beamsplitter parameters; e.g., see \cite{BR04} for a description of an NDPA. In this paper, we consider the issue of whether a similar experimental implementation may be provided for the distributed direct coupled quantum observer proposed in \cite{PET14Ca}.

The paper \cite{PET14Ca} proposes a direct coupled distributed quantum observer which is constructed via the direct connection of many quantum harmonic oscillators in a chain as illustrated in Figure \ref{F0}. It is shown that this quantum network can be constructed so that each output of the direct coupled distributed quantum observer converges to the plant output of interest in a time averaged sense. This is a form of time averaged quantum consensus for the quantum networks under consideration. However, the experimental implementation approach of \cite{PeHun1a,PeHun2a} cannot be extended in a straightforward way to the direct coupled distributed quantum observer \cite{PET14Ca}. This is because it is not feasible to extend the NDPA used in 
\cite{PeHun1a,PeHun2a} to allow for the multiple direct couplings to the multiple observer elements required in the theory of \cite{PET14Ca}. Hence, in this paper, we modify the theory of \cite{PET14Ca} to develop a new direct coupled distributed observer in which there is direct coupling only between the plant and the first element of the observer. All of the other couplings between the different elements of the observer are via optical field couplings. This is illustrated in Figure \ref{F1}.  Also, all of the elements of the observer except for the first one are implemented as passive optical cavities. The only active element in the augmented plant observer system is a single NDPA used to implement the plant and first observer element. These features mean that the proposed direct coupling observer is much easier to implement experimentally that the observer which was proposed in \cite{PET14Ca}. 

\begin{figure}[htbp]
\begin{center}
\includegraphics[width=8cm]{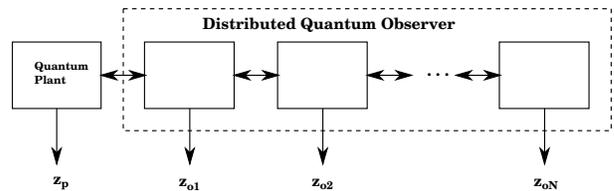}
\end{center}
\caption{Distributed Quantum Observer of \cite{PET14Ca}.}
\label{F0}
\end{figure}

\begin{figure*}[htbp]
\begin{center}
\includegraphics[width=16cm]{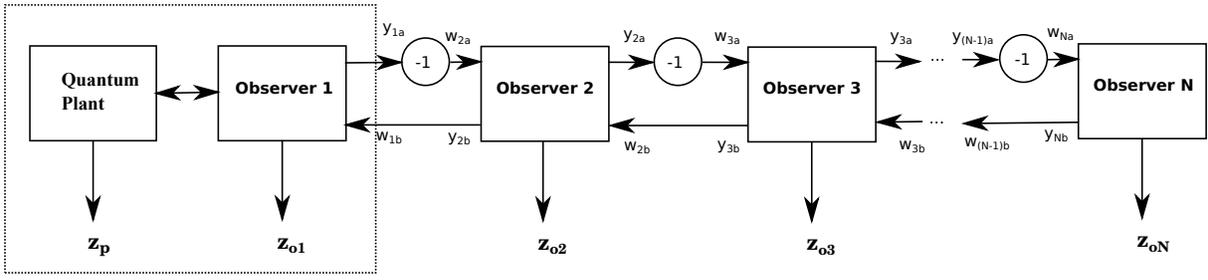}
\end{center}
\caption{Distributed Quantum Observer Proposed in This Paper.}
\label{F1}
\end{figure*}

We establish that the distributed quantum observer proposed in this paper has very similar properties to the distributed quantum observer proposed in \cite{PET14Ca} in that each output of the distributed observer converges to the plant output of interest in a time averaged sense. However, an important difference between the observer proposed in  \cite{PET14Ca} and the observer proposed in this paper is that in \cite{PET14Ca} the output for each observer element corresponded to the same quadrature whereas in this paper, different quadratures are used to define the outputs with a $90^\circ$ phase rotation as we move from observer element to element along the chain of observers.

\section{Quantum  Linear Systems}
In the distributed quantum observer problem under consideration, both the quantum plant and the distributed quantum observer are linear quantum systems; see also \cite{JNP1,GJ09,ZJ11}. 
The quantum mechanical behavior of a linear quantum system is described in terms of the system \emph{observables} which are self-adjoint operators on an underlying infinite dimensional complex Hilbert space $\mathfrak{H}$.   The commutator of two scalar operators $x$ and $y$ on ${\mathfrak{H}}$ is  defined as $[x, y] = xy - yx$.~Also, for a  vector of operators $x$ on ${\mathfrak H}$, the commutator of ${x}$ and a scalar operator $y$ on ${\mathfrak{H}}$ is the  vector of operators $[{x},y] = {x} y - y {x}$, and the commutator of ${x}$ and its adjoint ${x}^\dagger$ is the  matrix of operators 
\[ [{x},{x}^\dagger] \triangleq {x} {x}^\dagger - ({x}^\# {x}^T)^T, \]
where ${x}^\# \triangleq (x_1^\ast\; x_2^\ast \;\cdots\; x_n^\ast )^T$ and $^\ast$ denotes the operator adjoint. 

The dynamics of the closed linear quantum systems under consideration are described by non-commutative differential equations of the form
\begin{eqnarray}
\dot x(t) &=& Ax(t); \quad x(0)=x_0
 \label{quantum_system}
\end{eqnarray}
where $A$ is a real matrix in $\rbb^{n
\times n}$, and $ x(t) = [\begin{array}{ccc} x_1(t) & \ldots &
x_n(t)
\end{array}]\trp$ is a vector of system observables; e.g., see \cite{JNP1}. Here $n$ is assumed to be an even number and $\frac{n}{2}$ is the number of modes in the quantum system. 

The initial system variables $x(0)=x_0$ 
are assumed to satisfy the {\em commutation relations}
\begin{equation}
[x_j(0), x_k(0) ] = 2 \imath \Theta_{jk}, \ \ j,k = 1, \ldots, n,
\label{x-ccr}
\end{equation}
where $\Theta$ is a real skew-symmetric matrix with components
$\Theta_{jk}$.  In the case of a
single quantum harmonic oscillator, we will choose $x=(x_1, x_2)^T$ where
$x_1=q$ is the position operator, and $x_2=p$ is the momentum
operator.  The
commutation relations are  $[q,p]=2 i$.
In general, the matrix $\Theta$ is assumed to be  of the  form
\begin{equation}
\label{Theta}
\Theta=\diag(J,J,\ldots,J)
\end{equation}
 where $J$ denotes the real skew-symmetric $2\times 2$ matrix
$$
J= \left[ \begin{array}{cc} 0 & 1 \\ -1 & 0
\end{array} \right].$$

The system dynamics (\ref{quantum_system}) are determined by the system Hamiltonian 
which is a self-adjoint operator on the underlying  Hilbert space  $\mathfrak{H}$. For the linear quantum systems under consideration, the system Hamiltonian will be a
quadratic form
$\mathcal{H}=\half x(0)\trp R x(0)$, where $R$ is a real symmetric matrix. Then, the corresponding matrix $A$ in 
(\ref{quantum_system}) is given by 
\begin{equation}
A=2\Theta R \label{eq_coef_cond_A}.
\end{equation}
 where $\Theta$ is defined as in (\ref{Theta});
e.g., see \cite{JNP1}.
In this case, the  system variables $x(t)$ 
will satisfy the {\em commutation relations} at all times:
\begin{equation}
\label{CCR}
[x(t),x(t)^T]=2\imath \Theta \ \mbox{for all } t\geq 0.
\end{equation}
That is, the system will be \emph{physically realizable}; e.g., see \cite{JNP1}.

\begin{remark}
\label{R1}
Note that that the Hamiltonian $\mathcal{H}$ is preserved in time for the system (\ref{quantum_system}). Indeed,
$ \mathcal{\dot H} = \frac{1}{2}\dot{x}^TRx+\frac{1}{2}x^TR\dot{x} = -x^TR\Theta R x + x^TR\Theta R x = 0$ since $R$ is symmetric and $\Theta$ is skew-symmetric.
\end{remark}

\section{Direct Coupling Distributed Coherent Quantum Observers}
In our proposed direct coupling coherent quantum observer, the quantum plant is a single quantum harmonic oscillator which is a linear quantum system of the form (\ref{quantum_system}) described by the non-commutative differential equation
\begin{eqnarray}
\dot x_p(t) 
&=& A_px_p(t); \quad x_p(0)=x_{0p}; \nonumber \\
z_p(t) &=& C_px_p(t)
 \label{plant}
\end{eqnarray}
where $z_p(t)$ denotes the vector of system variables to be estimated by the observer and $ A_p \in \rbb^{2
\times 2}$, $C_p\in \rbb^{1 \times 2}$. 
It is assumed that this quantum plant corresponds to a plant Hamiltonian
$\mathcal{H}_p=\half x_p(0)\trp R_p x_p(0)$. Here $x_p = \left[\begin{array}{l}q_p\\p_p\end{array}\right]$ where
$q_p$ is the plant position operator and $p_p$ is the plant momentum operator. As in \cite{PET14Ca}, in the sequel we will assume that $A_p= 0$. 

We now describe the linear quantum system of the form (\ref{quantum_system}) which will correspond to the distributed quantum observer; see also \cite{JNP1,GJ09,ZJ11}. 
This system is described by a non-commutative differential equation of the form
\begin{eqnarray}
\dot x_o(t) &=& A_ox_o(t);\quad x_o(0)=x_{0o};\nonumber \\
z_o(t) &=& C_ox_o(t)
 \label{observer}
\end{eqnarray}
where the observer output $z_o(t)$ is the distributed observer estimate vector and $ A_o \in \rbb^{n_o
\times n_o}$, $C_o\in \rbb^{\frac{n_o}{2} \times n_o}$.  Also,  $x_o(t)$  is a vector of self-adjoint 
non-commutative system variables; e.g., see \cite{JNP1}. We assume the distributed observer order $n_o$  is an even number with $N=\frac{n_o}{2}$ being the number of elements in the distributed quantum observer. We also assume that the plant variables commute with the observer variables. 
% The system dynamics (\ref{observer}) are determined by the observer system Hamiltonian 
% which is a self-adjoint operator on the underlying  Hilbert space for the observer. 
% For the distributed quantum observer under consideration, this Hamiltonian is given by a 
% quadratic form:
% $\mathcal{H}_o=\half x_o(0)\trp R_o x_o(0)$, where $R_o$ is a real symmetric matrix. Then, the corresponding matrix $A_o$ in 
% (\ref{observer}) is given by 
% \begin{equation}
% A_o=2\Theta R_o \label{eq_coef_cond_Ao}
% \end{equation}
%  where $\Theta$ is defined as in (\ref{Theta}). Furthermore,
We will assume that the distributed quantum observer has a chain structure and is coupled to the quantum plant as shown in Figure \ref{F1}. 
Furthermore, we write
\[
z_o = \left[\begin{array}{l}z_{o1}\\z_{o2}\\\vdots\\z_{oN}\end{array}\right]
\]
where 
\[
z_{oi} = C_{oi}x_{oi} \mbox{ for }i=1,2,\ldots,N.
\]
Note that  $C_{oi} \in \rbb^{1 \times 2}$.
%, and each matrix $R_{oi}$ is symmetric for $i=1,2,\ldots,N$.

The augmented quantum linear system consisting of the quantum plant (\ref{plant})  and the distributed  quantum observer (\ref{observer}) is then a quantum system of the form (\ref{quantum_system}) described by equations of the form 
where
\begin{eqnarray}
\left[\begin{array}{l}\dot x_p(t)\\\dot x_{o1}(t)\\\dot x_{o2}(t)\\\vdots\\\dot x_{oN}(t)\end{array}\right] &=& 
A_a\left[\begin{array}{l} x_p(t)\\ x_{o1}(t)\\x_{o2}(t)\\\vdots\\x_{oN}(t)\end{array}\right];\nonumber \\
z_p(t) &=& C_px_p(t);\nonumber \\
z_o(t) &=& C_ox_o(t)
\label{augmented_system}
\end{eqnarray}
where 
\[
C_o = \left[\begin{array}{llll}C_{o1} & & &\\
                               & C_{o2}& 0 &\\
                               & 0 & \ddots & \\
                               &&& C_{oN}
\end{array}\right].
\]

We now formally define the notion of a direct coupled linear quantum observer.

\begin{definition}
\label{D1}
The  {\em distributed linear quantum observer} (\ref{observer}) is said to achieve time-averaged consensus convergence for the quantum plant (\ref{plant}) if the corresponding augmented linear quantum system  (\ref{augmented_system}) is such that
\begin{equation}
\label{average_convergence}
\lim_{T \rightarrow \infty} \frac{1}{T}\int_{0}^{T}(\left[\begin{array}{l}1\\1\\\vdots\\1\end{array}\right]z_p(t) 
- z_o(t))dt = 0.
\end{equation}
\end{definition}
\section{Implementation of a  Distributed Quantum Observer} \label{sec:intro}

We will consider a distributed quantum observer which has a chain structure and is coupled to the quantum plant as shown in Figure \ref{F1}. 
In this distributed quantum observer, there a direct coupling between the quantum plant and the first quantum observer. This direct coupling is determined by a coupling Hamiltonian which defines the coupling between the quantum plant and the first element of the distributed quantum observer:
\begin{equation}
\label{Rc}
\mathcal{H}_c = x_{p}(0)\trp R_{c} x_{o1}(0).
\end{equation}
However, in contrast to \cite{PET14Ca}, there is field coupling between the first quantum observer and all other quantum observers in the chain of observers. The motivation for this structure is that it would be much easier to implement experimentally than the structure proposed in \cite{PET14Ca}. Indeed, the subsystem consisting of the quantum plant and the first quantum observer can be implemented using an NDPA and a beamsplitter in a similar way to that described in \cite{PeHun1a,PeHun2a}; see also \cite{BR04} for further details on NDPAs and beamsplitters. This is illustrated in Figure \ref{F2}.

\begin{figure}[htbp]
\begin{center}
\includegraphics[width=8cm]{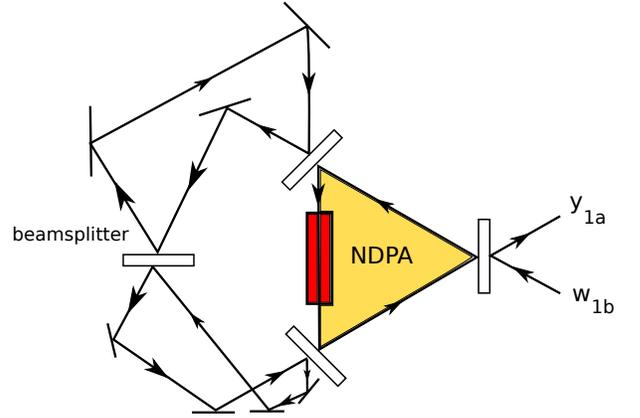}
\end{center}
\caption{NDPA coupled to a beamsplitter representing the quantum plant and first quantum observer.}
\label{F2}
\end{figure}

Also, the remaining quantum observers in the distributed quantum observer are implemented as simple cavities as shown in Figure \ref{F3}. 

\begin{figure}[htbp]
\begin{center}
\includegraphics[width=8cm]{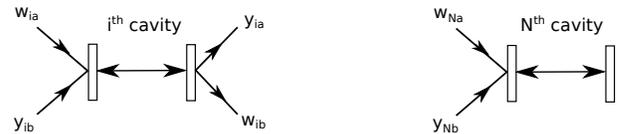}
\end{center}
\caption{Optical cavity implementation of the remaining quantum observers in the distributed quantum observer.}
\label{F3}
\end{figure}

The proposed quantum optical implementation of a distributed quantum observer is simpler than that of \cite{PET14Ca}. However, its dynamics are somewhat different than those of the distributed quantum observer proposed in \cite{PET14Ca}. We now proceed to analyze these dynamics. Indeed, using the results of \cite{PeHun2a}, we can write down quantum stochastic differential equations  (QSDEs) describing the plant-first observer system shown in Figure \ref{F2}:
\begin{eqnarray}
\label{augmented1}
d x_p &=& 2 J \alpha\beta^T x_{o1} dt;\nonumber \\
d x_{o1} &=& 2\omega_1 J x_{o1} dt  -\frac{1}{2}\kappa_{1b} x_{o1} dt + 2J\beta\alpha^T x_p dt\nonumber \\
&& - \sqrt{\kappa_{1b}}dw_{1b}; \nonumber \\
dy_{1a} &=& \sqrt{\kappa_{1b}} x_{o1} dt + dw_{1b}
\end{eqnarray}
where $x_p = \left[\begin{array}{c} q_p \\ p_p \end{array}\right]$ is the vector of position and momentum operators for the quantum plant and $x_{o1} = \left[\begin{array}{c} q_1 \\ p_1 \end{array}\right]$ is the vector of position and momentum operators for the first quantum observer. 
Here, $\alpha \in \rbb^2$, $\beta  \in \rbb^2$ and  $\kappa_{1b} > 0$ are parameters which depend on the parameters of the beamsplitter and the NDPA. The parameters $\alpha$ and $\beta$ define the coupling Hamiltonian matrix defined in (\ref{Rc}) as follows:
\begin{equation}
\label{Rc_def}
R_c = \alpha\beta^T.
\end{equation}
In addition, the parameters of the beamsplitter and the NDPA need to be chosen as described in  
\cite{PeHun1a,PeHun2a} in order to obtain QSDEs of the required form (\ref{augmented1}). 

The QSDEs describing the $ith$ quantum observer for $i=2,3,\ldots,N-1$ are as follows:
\begin{eqnarray}
\label{cavity_i}
d x_{oi} &=& 2\omega_i J x_{oi} dt  -\frac{\kappa_{ia}+\kappa_{ib}}{2} x_{oi} dt\nonumber \\
&&- \sqrt{\kappa_{ia}}dw_{ia}- \sqrt{\kappa_{ib}}dw_{ib}; \nonumber \\
dy_{ia} &=& \sqrt{\kappa_{ib}} x_{oi} dt + dw_{ib}; \nonumber \\
dy_{ib} &=& \sqrt{\kappa_{ia}} x_{oi} dt + dw_{ia}
\end{eqnarray}
where  $x_{oi} = \left[\begin{array}{c} q_i \\ p_i \end{array}\right]$ is the vector of position and momentum operators for the $ith$ quantum observer; e.g., see \cite{BR04}. Here $\kappa_{ia} >0$ and $\kappa_{ib} > 0$ are parameters relating to the reflectivity of each of the partially reflecting mirrors which make up the cavity. 

The QSDEs describing the $Nth$ quantum observer  are as follows:
\begin{eqnarray}
\label{cavity_N}
d x_{oN} &=& 2\omega_N J x_{oN} dt  -\frac{\kappa_{Na}}{2} x_{oN} dt - \sqrt{\kappa_{Na}}dw_{Na}; \nonumber \\
dy_{Nb} &=& \sqrt{\kappa_{Na}} x_{oN} dt + dw_{Na}
\end{eqnarray}
where  $x_{oN} = \left[\begin{array}{c} q_N \\ p_N \end{array}\right]$ is the vector of position and momentum operators for the $Nth$ quantum observer. Here $\kappa_{Na} >0$ is a parameter relating to the reflectivity of the partially reflecting mirror in  this cavity. 

In addition to the above equations, we also have the following equations which describe the interconnections between the observers as in Figure \ref{F1}:
\begin{eqnarray}
\label{interconnection_i}
w_{(i+1)a} &=& - y_{ia}; \nonumber \\
w_{ib} &=& y_{(i+1)b}
\end{eqnarray}
for $i = 1,2,\ldots,N-1$. 

In order to describe the augmented system consisting of the quantum plant and the quantum observer, we now combine equations (\ref{augmented1}), (\ref{cavity_i}), (\ref{cavity_N})  and (\ref{interconnection_i}). Indeed, starting with observer $N$, we have from  (\ref{cavity_N}), (\ref{interconnection_i})
\[
dy_{Nb} = \sqrt{\kappa_{Na}} x_{oN} dt - dy_{(N-1)a}.
\]
But from (\ref{cavity_i}) with $i= N-1$,
\[
dy_{(N-1)a} = \sqrt{\kappa_{(N-1)b}} x_{o(N-1)} dt + dw_{(N-1)b}.
\]
Therefore,
\begin{eqnarray*}
dy_{Nb} &=& \sqrt{\kappa_{Na}} x_{oN} dt - \sqrt{\kappa_{(N-1)b}} x_{o(N-1)} dt - dw_{(N-1)b}\nonumber \\
&=&  \sqrt{\kappa_{Na}} x_{oN} dt - \sqrt{\kappa_{(N-1)b}} x_{o(N-1)} dt - dy_{Nb}
\end{eqnarray*}
using (\ref{interconnection_i}). Hence, 
\begin{equation}
\label{yNb}
dy_{Nb} = \frac{\sqrt{\kappa_{Na}}}{2} x_{oN} dt - \frac{\sqrt{\kappa_{(N-1)b}}}{2} x_{o(N-1)} dt. 
\end{equation}
From this, it follows using (\ref{cavity_N}) that
\begin{eqnarray*}
dw_{Na} &=& -\sqrt{\kappa_{Na}} x_{oN} dt + dy_{Nb}\nonumber \\
&=& -\frac{\sqrt{\kappa_{Na}}}{2} x_{oN} dt - \frac{\sqrt{\kappa_{(N-1)b}}}{2} x_{o(N-1)} dt. 
\end{eqnarray*}
Then,  using  (\ref{cavity_N}) we obtain the equation
\begin{equation}
\label{xN}
d x_{oN} = 2\omega_N J x_{oN} dt  +  \frac{\sqrt{\kappa_{(N-1)b}\kappa_{Na}}}{2} x_{o(N-1)} dt.
\end{equation}

We now consider observer $N-1$. Indeed, it follows from (\ref{cavity_i}) and (\ref{interconnection_i}) with $i = N-1$ that
\begin{eqnarray}
\label{xN-1a}
d x_{o(N-1)} &=& 2\omega_{N-1} J x_{o(N-1)} dt  \nonumber \\
&&-\frac{\kappa_{(N-1)a}+\kappa_{(N-1)b}}{2} x_{o(N-1)} dt\nonumber \\
&&- \sqrt{\kappa_{(N-1)a}}dw_{(N-1)a}- \sqrt{\kappa_{(N-1)b}}dy_{Nb} \nonumber \\
 &=& 2\omega_{N-1} J x_{o(N-1)} dt  \nonumber \\
&& -\frac{\kappa_{(N-1)a}+\kappa_{(N-1)b}}{2} x_{o(N-1)} dt\nonumber \\
&&- \sqrt{\kappa_{(N-1)a}}dw_{(N-1)a}\nonumber\\
&&- \frac{\sqrt{\kappa_{Na}\kappa_{(N-1)b}}}{2} x_{oN} dt \nonumber \\
&&+ \frac{\kappa_{(N-1)b}}{2} x_{o(N-1)} dt \nonumber \\
&=&  2\omega_{N-1} J x_{o(N-1)} dt   -\frac{\kappa_{(N-1)a}}{2} x_{o(N-1)} dt\nonumber \\
&&- \frac{\sqrt{\kappa_{Na}\kappa_{(N-1)b}}}{2} x_{oN} dt \nonumber \\
&&- \sqrt{\kappa_{(N-1)a}}dw_{(N-1)a}
\end{eqnarray}
using (\ref{yNb}). Now using (\ref{cavity_i}) and (\ref{interconnection_i}) with $i = N-2$, it follows that 
\begin{eqnarray*}
dy_{(N-2)a} &=& \sqrt{\kappa_{(N-2)b}} x_{o(N-2)} dt + dw_{(N-2)b} \nonumber \\
&=& \sqrt{\kappa_{(N-2)b}} x_{o(N-2)} dt + dy_{(N-1)b} \nonumber \\
&=& \sqrt{\kappa_{(N-2)b}} x_{o(N-2)} dt +  \sqrt{\kappa_{(N-1)a}} x_{o(N-1)} dt \nonumber \\
&&+ dw_{(N-1)a}
\end{eqnarray*}
using  (\ref{cavity_i}) with $i = N-1$. 
Hence using (\ref{interconnection_i}) with $i = N-2$, it follows that 
\begin{eqnarray*}
dy_{(N-2)a} &=& \sqrt{\kappa_{(N-2)b}} x_{o(N-2)} dt +  \sqrt{\kappa_{(N-1)a}} x_{o(N-1)} dt \nonumber \\
&&- dy_{(N-2)a}.
\end{eqnarray*}
Therefore 
\[
dy_{(N-2)a} = \frac{\sqrt{\kappa_{(N-2)b}}}{2} x_{o(N-2)} dt +  \frac{\sqrt{\kappa_{(N-1)a}}}{2} x_{o(N-1)} dt.
\]
Substituting this into (\ref{xN-1a}), we obtain 
\begin{eqnarray}
\label{xN-1}
d x_{o(N-1)} &=& 2\omega_{N-1} J x_{o(N-1)} dt  \nonumber \\
&&- \frac{\sqrt{\kappa_{(N-1)b}\kappa_{Na}}}{2} x_{oN} dt \nonumber \\
&&+ \frac{\sqrt{\kappa_{(N-2)b}\kappa_{(N-1)a}}}{2} x_{o(N-2)} dt.
\end{eqnarray}

Continuing this process, we obtain the following QSDEs for the variables $x_{oi}$: 
\begin{eqnarray}
\label{xi}
d x_{oi} &=& 2\omega_i J x_{oi} dt  \nonumber \\
&&- \frac{\sqrt{\kappa_{ib}\kappa_{(i+1)a}}}{2} x_{o(i+1)} dt \nonumber \\
&&+ \frac{\sqrt{\kappa_{(i-1)b}\kappa_{ia}}}{2} x_{o(i-1)} dt
\end{eqnarray}
for $i=2,3, \ldots, N-1$. Finally for $x_{o1}$, we obtain
\begin{eqnarray}
\label{x1}
d x_{o1} &=& 2\omega_1 J x_{o1} dt  - \frac{\sqrt{\kappa_{1b}\kappa_{(2a}}}{2} x_{o2} dt \ + 2J\beta\alpha^T x_p dt.\nonumber \\
\end{eqnarray}

We now observe that the plant equation
\begin{eqnarray}
\label{xp}
d x_p &=& 2 J \alpha\beta^T x_{o1} dt\nonumber \\
\end{eqnarray}
implies that the quantity
\[
z_p = \alpha^T x_p 
\]
satisfies 
\[
d z_p = 2 \alpha^T J \alpha\beta^T x_{o1} dt = 0
\]
since $J$ is a skew-symmetric matrix. Therefore, 
\begin{equation}
\label{zp_const}
z_p(t) = z_p(0) = z_p
\end{equation}
for all $t\geq 0$. 

We now combine equations (\ref{x1}), (\ref{xi}), (\ref{xN}) and write them in vector-matrix form. Indeed, let 
\[
x_o = \left[\begin{array}{c}x_{o1}\\x_{o2}\\\vdots\\x_{oN}\end{array}\right].
\]
Then, we can write
\[
\dot x_o = A_o x_o + B_o z_p
\]
where 
{\small
\begin{eqnarray}
\label{Ao}
\lefteqn{A_o =}\nonumber \\
&&2\left[\begin{array}{rrrrrr} \omega_1 J & -\mu_2 I & & &  & \\
                                  \mu_2 I & \omega_2 J& -\mu_3 I & & 0 & \\
                                          &  \mu_3 I & \omega_3 J& -\mu_4 I && \\
                                          &  & \ddots & \ddots & \ddots  & \\
                                    0     &  &  & \mu_{N-1} I & \omega_{N-1} J & -\mu_N I\\
                                          &  &  &  & \mu_N I & \omega_N J
\end{array}\right], \nonumber \\
\lefteqn{B_o =}\nonumber \\
&& 2\left[\begin{array}{r}J \beta \\ 0 \\ \vdots \\ 0 \end{array}\right]
\end{eqnarray}}
and 
\[
\mu_i = \frac{1}{4} \sqrt{\kappa_{(i-1)b}\kappa_{ia}} > 0
\]
for $i = 2, 3, \ldots , N$. 

To construct a suitable distributed quantum observer, we will further assume that 
\begin{eqnarray}
\label{alphabeta}
\beta &=& -\mu_1 \alpha,\nonumber \\
C_p&=& \alpha^T,
\end{eqnarray}
 where  $\mu_1 > 0$
and
{\small
\begin{eqnarray}
\label{Co}
\lefteqn{C_{o} =}\nonumber \\
 &&\frac{1}{\|\alpha\|^2}\left[\begin{array}{llllll}\alpha^T &&&&&\\
                                & -J\alpha^T &&& 0 &\\
                                && -\alpha^T &&&\\
                                &&& J\alpha^T && \\
                                &0&&& \ddots & \\
                                &&&&& (-J)^{N-1}\alpha^T
\end{array}\right].\nonumber \\
\end{eqnarray}}
This choice of the matrix $C_o$ means that different quadratures are used for the outputs of the elements of the distributed quantum observer with a $90^\circ$ phase rotation as we move from observer element to element along the chain of observers.

In order to construct suitable values for the quantities $\mu_i$ and $\omega_i$, we require that 
\begin{equation}
\label{xobareqn}
A_o\bar x_o +B_oz_p = 0
\end{equation}
where
\[
\bar x_o = \left[\begin{array}{c}\alpha\\J\alpha\\-\alpha \\ -J\alpha \\ \alpha \\ \vdots\\(J)^{N-1}\alpha\end{array}\right]z_p.
\]
This will ensure that the quantity
\begin{equation}
\label{xe}
 x_e = x_o - \bar x_o
\end{equation}
 will satisfy the non-commutative differential equation
\begin{equation}
\label{xedot}
\dot{x}_e = A_o  x_e.
\end{equation}
This, combined with the fact that
\begin{eqnarray}
\label{Coxbar}
C_o\bar{x}_o&=& \frac{1}{\|\alpha\|^2}\left[\begin{array}{llll}\alpha^T &&&\\
                                & -J\alpha^T & 0 &\\
                                &0& \ddots & \\
                                &&& (-J)^{N-1}\alpha^T
\end{array}\right]\nonumber \\
&& \times \left[\begin{array}{l}\alpha\\J\alpha\\ \vdots\\(J)^{N-1}\alpha\end{array}\right]z_p \nonumber \\
&=& \left[\begin{array}{l}1\\1\\\vdots\\1\end{array}\right]z_p
\end{eqnarray}
will be used in establishing   condition (\ref{average_convergence}) 
for the distributed quantum observer.

Now, we require
\begin{eqnarray*}
\lefteqn{A_o\left[\begin{array}{c}\alpha\\J\alpha\\-\alpha \\ -J\alpha \\ \alpha \\ \vdots\\(J)^{N-1}\end{array}\right]+B_o}\nonumber \\
&=&
2\left[\begin{array}{c}
\omega_1J\alpha-\mu_2J\alpha-\mu_1J\alpha\\
\mu_2 \alpha- \omega_2\alpha+ \mu_3\alpha\\ 
\mu_3J\alpha-\omega_3J\alpha+\mu_4J\alpha \\
\vdots\\
\mu_N(J)^{N-2} \alpha+\omega_NJ^N\alpha
\end{array}\right]\nonumber \\
&=& 0.
\end{eqnarray*}
This will be satisfied if and only if 
\[
\left[\begin{array}{c}
\omega_1-\mu_2-\mu_1\\
\mu_2- \omega_2+ \mu_3\\ 
\mu_3- \omega_3+\mu_4 \\
\vdots\\
\mu_N-\omega_N
\end{array}\right] = 0.
\]
That is, we will assume that
\begin{equation}
\label{mui}
\omega_i=\mu_i+\mu_{i+1}
\end{equation}
for $i=1,2,\ldots,N-1$ and
\begin{equation}
\label{muN}
\omega_N=\mu_N. 
\end{equation}

To show that the above candidate distributed quantum observer leads to the satisfaction of the  condition
(\ref{average_convergence}), 
we first note that $x_e$ defined in (\ref{xe}) will satisfy (\ref{xedot}). If we can show that 
\begin{equation}
\label{xeav}
\lim_{T \rightarrow \infty} \frac{1}{T}\int_{0}^{T} x_e(t)dt = 0,
\end{equation}
 then it will follow from (\ref{Coxbar}) and (\ref{xe}) that (\ref{average_convergence}) is satisfied. In order to establish (\ref{xeav}), we first note that we can write
\[
A_o = 2\Theta R_o
\]
where
\begin{small}
\begin{eqnarray*}
\lefteqn{R_o = }\nonumber \\
&\left[\begin{array}{rrrrrr}\omega_1I & \mu_2J & & & &\\
 -\mu_2J & \omega_2I & \mu_3J&  & 0 & \\
&  -\mu_3J & \omega_3I & \mu_4J &&\\
& & \ddots & \ddots  & \ddots &\\
& 0 && -\mu_{N-1}J& \omega_{N-1}I & \mu_NJ\\
&&&&  -\mu_NJ & \omega_NI
\end{array}\right].
\end{eqnarray*}
\end{small}
We will now show that the symmetric matrix $R_o$ is positive-definite.

\begin{lemma}
\label{L1}
The matrix $R_o$ is positive definite.
\end{lemma}

\begin{proof}
In order to establish this lemma, let 
\[
x_o =  \left[\begin{array}{l}x_{o1}\\x_{o2}\\\vdots\\x_{oN}\end{array}\right] \in \rbb^{2N}
\]
where  $x_{oi} = \left[\begin{array}{c} q_i \\ p_i \end{array}\right] \in \rbb^2$ for $i = 1,2, \ldots, N$. Also, define the complex scalars $a_i = q_i + \imath p_i$ for $i=1,2,\ldots,N$. Then it is straightforward to verify that 
\begin{eqnarray*}
x_o^TR_ox_o &=& \omega_1\|x_{o1}\|^2-2\mu_2x_{o1}^T\alpha x_{o2}^T\alpha+\omega_2\|x_{o2}\|^2\nonumber \\
&& -2 \mu_3 x_{o2}^T\alpha x_{o3}^T\alpha+\omega_3\|x_{o3}\|^2\nonumber \\
&& \vdots \nonumber \\
&&-2 \mu_N x_{oN-1}^T\alpha x_{oN}^T\alpha+\omega_N\|x_{oN}\|^2\nonumber \\
&= & \omega_1a_1^*a_1-\imath\mu_2a_1^*a_2+\imath\mu_2a_2^*a_1+\omega_2a_2^*a_2\nonumber \\
&& -\imath \mu_3 a_2^*a_3+\imath \mu_3 a_3^*a_2+\omega_3a_3^*a_3\nonumber \\
&& \vdots \nonumber \\
&&-\imath \mu_N a_{N-1}^*a_N+\imath \mu_N a_N^*a_{N-1}+\omega_Na_N^*a_N\nonumber \\
&=& a_o^\dagger\tilde R_o a_o
\end{eqnarray*}
where 
\[
a_o = \left[\begin{array}{l}a_{1}\\a_{2}\\\vdots\\a_{N}\end{array}\right] \in \cbb^N
\]
and 
\[
 \tilde R_o= \left[\begin{array}{rrrrr}\omega_1 & -\imath \mu_2 & & &\\
 \imath \mu_2 & \omega_2 & -\imath \mu_3&  & 0 \\
&  \imath \mu_3 & \omega_3 & \ddots &\\
0 & & \ddots & \ddots  & -\imath \mu_N\\
&&&  \imath \mu_N & \omega_N
\end{array}\right]. 
\]
Here $^\dagger$ denotes the complex conjugate transpose of a vector. 
From this, it follows that the real symmetric matrix $R_o$  is positive-definite if and only if the complex Hermitian matrix $\tilde R_o$ is positive-definite. 

To prove that $\tilde R_o$ is positive-definite, we first substitute the equations (\ref{mui}) and (\ref{muN}) into the definition of $\tilde R_o$ to obtain
\begin{eqnarray*}
\tilde R_o&=& \left[\begin{array}{rrrrr}\mu_1+\mu_2 & -\imath \mu_2 & & &\\
 \imath \mu_2 & \mu_2+\mu_3 & -\imath \mu_3&  & 0 \\
&  \imath \mu_3 & \mu_3+\mu_4 & \ddots &\\
0 & & \ddots & \ddots  & -\imath \mu_N\\
&&&  \imath \mu_N & \mu_N
\end{array}\right]\nonumber \\
&=& \tilde R_{o1} + \tilde R_{o2}
\end{eqnarray*}
where
\[
\tilde R_{o1} = \left[\begin{array}{rrrr}
 \mu_1  & 0 &\ldots  & 0\\
0 & 0  &\ldots  & 0\\
\vdots & & & \vdots\\
0 & 0  &\ldots  & 0
\end{array}\right] \geq 0
\]
and
\[
\tilde R_{o2} = \left[\begin{array}{rrrrr} \mu_{2}  & - \imath \mu_2 & & &\\
 +\imath \mu_2 &  \mu_2+ \mu_{3} & -\imath \mu_3&  & 0 \\
&  +\imath\mu_3 &  \mu_3+ \mu_{4} & \ddots &\\
0 & & \ddots & \ddots  & - \imath \mu_N\\
&&&  +\imath \mu_N &  \mu_N
\end{array}\right].
\]

Now, we can write
\begin{eqnarray*}
a_o^\dagger\tilde R_{o2} a_o &=&
\mu_2 a_1^*a_1-\imath\mu_2a_1^*a_2+\imath\mu_2a_2^*a_1+\mu_2a_2^*a_2\nonumber \\
&&+\mu_3a_2^*a_2 -\imath \mu_3 a_2^*a_3+\imath \mu_3 a_3^*a_2+\mu_4a_3^*a_3\nonumber \\
&& \vdots \nonumber \\
&&+\mu_{N-1}a_{N-1}^*a_{N-1}-\imath \mu_N a_{N-1}^*a_N\nonumber \\
&&+\imath \mu_N a_N^*a_{N-1}+\mu_Na_N^*a_N\nonumber \\
&=& \mu_2 (-\imath a_1^* + a_2^*)(\imath a_1 + a_2)\nonumber \\ 
&&+\mu_3 (-\imath a_2^* + a_3^*)(\imath a_2 + a_3)\nonumber \\ 
&&\vdots \nonumber \\
&&+\mu_N (-\imath a_{N-1}^* + a_N^*)(\imath a_{N-1} + a_N)\nonumber \\ 
&\geq& 0. 
\end{eqnarray*}
Thus, $\tilde R_{o2} \geq 0$. Furthermore, $a_o^\dagger\tilde R_{o2} a_o = 0$ if and only if 
\begin{eqnarray*}
a_2 &=& -\imath a_1;\nonumber \\
a_3 &=& -\imath a_2;\nonumber \\
&\vdots& \nonumber \\
a_N &=& -\imath a_{N-1}. \nonumber \\
\end{eqnarray*}
That is, the null space of $\tilde R_{o2}$ is given by 
\[
\mathcal{N}(\tilde R_{o2}) = \mbox{span}\{\left[\begin{array}{l}1\\-\imath\\ -1\\ \imath \\ 1 \\\vdots\\(-\imath)^{N-1}\end{array}\right]\}.
\]

The fact that $\tilde R_{o1} \geq 0$ and $\tilde R_{o2} \geq 0$ implies that $\tilde R_{o} \geq 0$. In order to show that $\tilde R_{o} > 0$, suppose that $a_o$ is a non-zero vector in $\mathcal{N}(\tilde R_{o})$. It follows that 
\[
a_o^\dagger \tilde R_{o}a_o = a_o^\dagger\tilde R_{o1}a_o+a_o^\dagger\tilde R_{o2}a_o = 0.
\]
Since $\tilde R_{o1} \geq 0$ and $\tilde R_{o2} \geq 0$, $a_o$ must be contained in the null space of $\tilde R_{o1}$ and the null space of $\tilde R_{o2}$. Therefore $a_o$ must be of the form
\[
a_o = \gamma \left[\begin{array}{l}1\\-\imath\\ -1\\ \imath \\ 1 \\\vdots\\(-\imath)^{N-1}\end{array}\right]
\]
where $\gamma \neq 0$. However, then
\[
a_o^\dagger \tilde R_{o1}a_o = \gamma^2 \tilde \mu_1 \neq 0
\]
and hence $a_o$ cannot be in the null space of $\tilde R_{o1}$. Thus, we can conclude that the matrix $\tilde R_{o}$ is positive definite and hence, the matrix  $R_{o}$ is positive definite. This completes the proof of the lemma. 
\end{proof}

We now verify that the condition (\ref{average_convergence}) is satisfied for the distributed  quantum observer under consideration. This proof follows along very similar lines to the corresponding proof given in \cite{PET14Ca}. We recall from Remark \ref{R1} that the quantity $\half  x_e(t)\trp R_o  x_e(t)$
remains constant in time for the linear system:
\[
\dot{ x}_e = A_o x_e= 2\Theta R_o  x_e.
\]
That is 
\begin{equation}
\label{Roconst}
\half  x_e(t) \trp R_o  x_e(t) = \half  x_e(0) \trp R_o  x_e(0) \quad \forall t \geq 0.
\end{equation}
However, $ x_e(t) = e^{2\Theta R_ot} x_e(0)$ and $R_o > 0$. Therefore, it follows from (\ref{Roconst}) that
\[
\sqrt{\lambda_{min}(R_o)}\|e^{2\Theta R_ot} x_e(0)\| \leq \sqrt{\lambda_{max}(R_o)}\| x_e(0)\|
\]\
for all $ x_e(0)$ and $t \geq 0$. Hence, 
\begin{equation}
\label{exp_bound}
\|e^{2\Theta R_ot}\| \leq \sqrt{\frac{\lambda_{max}(R_o)}{\lambda_{min}(R_o)}}
\end{equation}
for all $t \geq 0$.

Now since $\Theta $ and $R_o$ are non-singular,
\[
\int_0^Te^{2\Theta R_ot}dt = \half e^{2\Theta R_oT}R_o^{-1}\Theta ^{-1} - \half R_o^{-1}\Theta ^{-1}
\]
and therefore, it follows from (\ref{exp_bound}) that
\begin{eqnarray*}
\lefteqn{\frac{1}{T} \|\int_0^Te^{2\Theta R_ot}dt\|}\nonumber \\
 &=& \frac{1}{T} \|\frac{1}{2}e^{2\Theta R_oT}R_o^{-1}\Theta ^{-1} - \frac{1}{2}R_o^{-1}\Theta ^{-1}\|\nonumber \\
&\leq& \frac{1}{2T}\|e^{2\Theta R_oT}\|\|R_o^{-1}\Theta ^{-1}\| \nonumber \\
&&+ \frac{1}{2T}\|R_o^{-1}\Theta ^{-1}\|\nonumber \\
&\leq&\frac{1}{2T}\sqrt{\frac{\lambda_{max}(R_o)}{\lambda_{min}(R_o)}}\|R_o^{-1}\Theta ^{-1}\|\nonumber \\
&&+\frac{1}{2T}\|R_o^{-1}\Theta ^{-1}\|\nonumber \\
&\rightarrow & 0 
\end{eqnarray*}
as $T \rightarrow \infty$. Hence,  
\begin{eqnarray*}
\lefteqn{\lim_{T \rightarrow \infty} \frac{1}{T}\|\int_{0}^{T}  x_e(t)dt\| }\nonumber \\
&=& \lim_{T \rightarrow \infty}\frac{1}{T}\|\int_{0}^{T} e^{2\Theta R_ot} x_e(0)dt\| \nonumber \\
&\leq& \lim_{T \rightarrow \infty}\frac{1}{T} \|\int_{0}^{T} e^{2\Theta R_ot}dt\|\| x_e(0)\|\nonumber \\
&=& 0.
\end{eqnarray*}
This implies
\[
\lim_{T \rightarrow \infty} \frac{1}{T}\int_{0}^{T}  x_e(t)dt = 0
\]
and hence, it follows from (\ref{xe}) and (\ref{Coxbar}) that
\[
\lim_{T \rightarrow \infty} \frac{1}{T}\int_{0}^{T} z_o(t)dt = \left[\begin{array}{l}1\\1\\\vdots\\1\end{array}\right]z_p.
\]

Also, (\ref{zp_const}) implies 
\[
\lim_{T \rightarrow \infty} \frac{1}{T}\int_{0}^{T} z_p(t)dt = \left[\begin{array}{l}1\\1\\\vdots\\1\end{array}\right]z_p.
\]
Therefore, condition (\ref{average_convergence}) is satisfied. Thus, we have established the following theorem.

\begin{theorem}
\label{T1}
Consider a quantum plant of the form (\ref{plant}) where  $A_p = 0$. Then the distributed direct coupled quantum observer defined by equations (\ref{observer}),  (\ref{Rc}), (\ref{Rc_def}), (\ref{Ao}), (\ref{alphabeta}), (\ref{Co}), (\ref{mui}), (\ref{muN})  achieves time-averaged consensus convergence for this quantum plant.
\end{theorem}

% \bibliography{/home/irp/Bibliog/irpnew}
% \bibliographystyle{IEEEtran}

\end{document}